\documentclass[conference]{IEEEtran}
\usepackage[letterpaper, top=0.71in, left=0.625in, right=0.625in, bottom=1in]{geometry}
\IEEEoverridecommandlockouts

\usepackage{cite}

\usepackage{amsmath,amssymb,amsfonts,amsthm,mathtools}

\usepackage{graphicx}
\usepackage{epstopdf}
\usepackage{pgfplots}
\pgfplotsset{compat=1.18}

\usepackage{float}
\usepackage{stfloats} 

\ifCLASSOPTIONcompsoc
  \usepackage[caption=false,font=normalsize,labelfont=sf,textfont=sf]{subfig}
\else
  \usepackage[caption=false,font=footnotesize]{subfig}
\fi

\usepackage[ruled,vlined,linesnumbered]{algorithm2e}

\usepackage{booktabs}
\usepackage{tabularx}
\usepackage{multirow}
\usepackage{array}
\usepackage{diagbox}

\usepackage{tikz}
\usetikzlibrary{arrows,shapes,positioning}

\usepackage{enumitem}

\usepackage{bbm}   

\newcommand{\RN}[1]{\textup{\uppercase\expandafter{\romannumeral#1}}}

\usepackage{titlesec}
\newtheorem{Proposition}{Proposition}
\definecolor{note_fontcolor}{rgb}{0, 0.667969, 1}

\definecolor{editcolor}{rgb}{0.8,0,0} 

\begin{document}

\setlength{\textfloatsep}{2pt plus 3pt minus 2pt}
\setlength{\abovecaptionskip}{2pt plus 3pt minus 2pt} 
\setlength{\belowcaptionskip}{2pt plus 3pt minus 2pt} 
\setlength{\belowdisplayskip}{3pt plus 3pt minus 2pt}
\setlength{\belowdisplayshortskip}{5pt plus 3pt minus 2pt}
\setlength{\abovedisplayskip}{3pt plus 3pt minus 2pt}
\setlength{\abovedisplayshortskip}{5pt plus 3pt minus 2pt}
\setlength{\skip\footins}{8pt plus 2pt minus 1pt}

\title{Secure Rate-Splitting and RIS Beamforming with Untrusted Energy Harvesting Receivers}

\author{\IEEEauthorblockN{Hamid Reza Hashempour\IEEEauthorrefmark{1}, Le-Nam Tran\IEEEauthorrefmark{2}, Duy H. N. Nguyen\IEEEauthorrefmark{3}, and Hien~Quoc~Ngo\IEEEauthorrefmark{1}}

\IEEEauthorblockA{
\IEEEauthorrefmark{1} Centre for Wireless Innovation (CWI), Queen’s University Belfast, U.K. \\
\IEEEauthorrefmark{2} School of Electrical and Electronic Engineering,
University College Dublin, Dublin, Ireland. \\
\IEEEauthorrefmark{3} Department of Electrical and Computer Engineering, San
Diego State University, San Diego, CA 92182, USA. \\
\{h.hashempoor, hien.ngo\}@qub.ac.uk, nam.tran@ucd.ie, duy.nguyen@sdsu.edu}}

\maketitle

\begin{abstract}
We consider a reconfigurable intelligent surface (RIS)-assisted heterogeneous network comprising legitimate information-harvesting receivers (IHRs) and untrusted energy-harvesting receivers (UEHRs). A multi-antenna base station (BS) transmits confidential information to IHRs while ensuring sufficient energy transfer to UEHRs that may attempt eavesdropping. To enhance physical-layer security, we propose a secure rate-splitting multiple access (RSMA) scheme aided by a UAV-mounted RIS. The objective is to maximize fairness-based secrecy energy efficiency (SEE). Owing to the non-convexity of the formulated problem, we develop an alternating optimization framework that jointly designs the common message allocation, active precoders, and RIS phase shifts under transmit power and energy harvesting constraints, leveraging sequential convex approximation (SCA). Simulation results demonstrate the scalability of the proposed algorithm and its superior SEE performance compared to space-division multiple access (SDMA) and non-orthogonal multiple access (NOMA) benchmarks.
\let\thefootnote\relax\footnotetext{ 
The work of H. R. Hashempour and H. Q. Ngo was supported by a research grant from the Department for the Economy Northern Ireland under the US-Ireland R\&D Partnership Programme.}
\end{abstract}


\IEEEpeerreviewmaketitle

\section{Introduction}\label{intro}
Unmanned aerial vehicles (UAVs) and reconfigurable intelligent surfaces (RISs) have emerged as promising enablers for coverage extension, energy-efficient operation, and physical-layer security in beyond-5G networks \cite{Survey,Hashempour0}. By repositioning UAV-mounted RISs, networks can establish favorable cascaded links in blockage-rich deployments, thereby strengthening desired signals and mitigating interference. Recent studies have highlighted both the opportunities and challenges of securing UAV–RIS systems. For instance, \cite{Mughal} employs deep machine learning to address malicious threats in UAV-assisted RIS networks, while \cite{Li-Renzo} investigate RIS-assisted secure communication where a UAV communicates with a ground user under the threat of eavesdropping, using time division multiple access (TDMA). Similarly, \cite{Shang} proposes a RIS-assisted secure UAV communication scheme in the presence of active jammers and passive eavesdroppers. 
In parallel, rate-splitting multiple access (RSMA) has been recognized as a robust strategy that enhances spectral efficiency and interference management by splitting each user’s message into common and private parts \cite{Bastami}. Its flexibility makes RSMA particularly attractive for security-constrained wireless systems \cite{Bastami,Hashempour1,Hashempour}.

Moreover, radio frequency (RF) signals inherently deliver both information and energy, motivating research on simultaneous wireless information and power transfer (SWIPT). RIS-assisted SWIPT systems have been studied for energy harvesting (EH) receivers \cite{Zhang}, while secure SWIPT has drawn attention due to potential eavesdropping by untrusted EH receivers. For instance, \cite{Ouyang} proposes artificial-noise-aided beamforming for full-duplex relays with untrusted EH receivers, and \cite{Mohammadi} studies RIS-assisted SWIPT in IoT networks. Recently, reinforcement learning has also been applied to maximize secrecy energy efficiency (SEE) in RIS-assisted networks \cite{Niyato}.

Motivated by these works, we study a secure multiuser multiple-input
single-output (MISO) downlink, where a base station (BS) serves multiple legitimate information-harvesting receivers (IHRs) in the presence of untrusted energy-harvesting receivers (UEHRs). To maximize fairness-based SEE while satisfying nonlinear EH constraints at UEHRs, we consider a RIS-assisted RSMA framework and jointly optimize the RSMA precoding, common message parameters, and RIS phase shifts. To the best of our knowledge, this joint design for SEE in RSMA-SWIPT systems has not been addressed before. The resulting non-convex problem is solved via an alternating optimization (AO) algorithm, where the common rate allocation is obtained through linear programming, and the BS precoders and RIS phase shifts via sequential convex approximation (SCA). Due to page limitations, the UAV-mounted RIS is assumed to be deployed at an effective location, and UAV positioning optimization is not considered in this work.
The main contributions of this work are summarized as follows:
\begin{itemize}
\item We study a secure RIS-assisted multiuser MISO downlink with SWIPT, where UEHRs can act as eavesdroppers, and fairness is ensured via a max--min SEE formulation.
\item We propose an efficient AO-based solution using Dinkelbach’s transform and SCA to jointly optimize RSMA precoding, rate allocation, and RIS phase shifts under nonlinear EH constraints.
\item Numerical results demonstrate that RSMA provides better robustness and scalability than SDMA and NOMA across various system settings.
\end{itemize}

The rest of the paper is organized as follows: Section~\ref{Sys_Model} presents the system model and RSMA design. Section~\ref{Problem-formulation} formulates the fairness SEE problem and details the AO framework. Section~\ref{Simulat} provides numerical results, and Section~\ref{conc} concludes the paper.

\section{System Model}\label{Sys_Model}
Consider a downlink wireless communication system consisting of one UAV-assisted RIS, $J$ non-colluding UEHRs, $K$ legitimate users, and one BS. The set of $K$ IHRs and $J$ UEHRs are denoted by $\mathcal{K} = \{1, \dots, K\}$ and $\mathcal{J} = \{1, \dots, J\}$, respectively. Due to the presence of high obstacles, there exists no line-of-sight (LoS) channel between the BS and the users, as well as between the BS and the UEHRs.
The UAV, equipped with a RIS, serves as a passive relay to assist communication between the BS and the users. The RIS is composed of $M$ reflecting elements, and the phase shift of each element can be controlled by the UAV.
The BS is equipped with $N_t$ antennas.
All UEHRs, IHRs, and the BS are located on the ground. The horizontal coordinates of user $k \in \mathcal{K}$, UEHR $j \in \mathcal{J}$, and the BS are denoted by $\mathbf{w}^I_k = [x_k^I, y_k^I]^T$, $\mathbf{w}^U_j = [x_j^U, y_j^U]^T$, and $\mathbf{w}_b = [x_b, y_b]^T$, respectively. 
The UAV flies at a fixed altitude $H$ and hovers to enhance communication. The system model is demonstrated in Fig. \ref{fig1}.

\begin{figure} 
	\centering
\includegraphics[width=.83\linewidth]{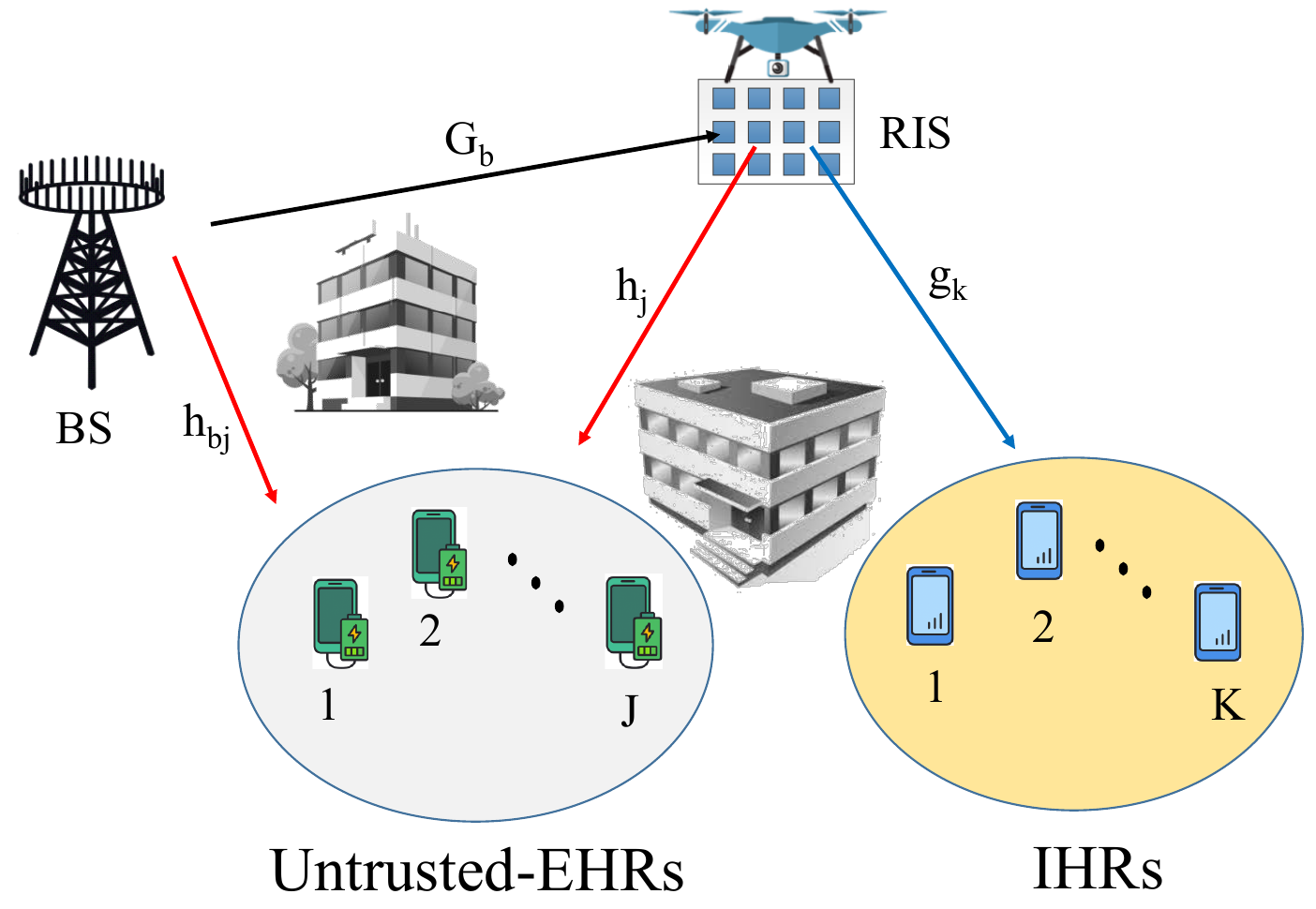}
	\caption{System model of UAV-assisted RIS-RSMA network.}
	\label{fig1}		
\end{figure}
The UAV horizontal coordinate is denoted by $\mathbf{q} = [x,y]^T$. 
The channel between the UAV and user $k$ is
\begin{equation}
\mathbf{g}_k = \tilde{\mathbf{g}}_k\, d_k^{-\frac{\alpha_k}{2}}, \quad 
\tilde{\mathbf{g}}_k = \sqrt{\tfrac{K_k}{K_k+1}}\,\mathbf{g}^{\mathrm{LoS}}_{k}
+ \sqrt{\tfrac{1}{K_k+1}}\,\hat{\mathbf{g}}_{k},
\end{equation}
where $\alpha_k\geq 2$ is the pathloss exponent, $\hat{\mathbf{g}}_{k} \sim \mathcal{CN}(\mathbf{0},\mathbf{I}_M)$, $K_k$ the Rician $K$-factor 
($K_k=0$ for Rayleigh fading), and $d_k=\sqrt{\|\mathbf{q}-\mathbf{w}^I_k\|^2+H^2}$. 
Similarly, the UAV–UEHR $j$ channel is
\begin{equation}
\mathbf{h}_j = \tilde{\mathbf{h}}_j\, d_j^{-\frac{\alpha_j}{2}}, \quad 
\tilde{\mathbf{h}}_j = \sqrt{\tfrac{K_j}{K_j+1}}\,\mathbf{h}^{\mathrm{LoS}}_{j}
+ \sqrt{\tfrac{1}{K_j+1}}\,\hat{\mathbf{h}}_{j},
\end{equation}
where $\hat{\mathbf{h}}_{j}\sim \mathcal{CN}(\mathbf{0},\mathbf{I}_M)$, $K_j$ is the Rician factor 
($K_j=0$ for Rayleigh fading), and $d_j=\sqrt{\|\mathbf{q}-\mathbf{w}^U_j\|^2+H^2}$.

We define the RIS phase shift vector  as $\boldsymbol{s} = [s_1, \dots, s_M]^T$, where each element is given by $s_m = e^{j \theta_m}$ for all $m \in \{1, \dots, M\}$.
The reflection-coefficient matrix of the RIS is then expressed as
\begin{equation}
\boldsymbol{\Theta} = \mathrm{diag}\left(s_1, \dots, s_M\right) \in \mathbb{C}^{M \times M}, \label{eq:reflection_matrix}
\end{equation}
where $\theta_m \in [0, 2\pi]$ denotes the phase shift applied by the $m$-th reflecting element.
The channel matrix between the BS and the RIS is denoted by $\mathbf{G}_b \in \mathbb{C}^{M \times N_t}$, where  $[\mathbf{G}_b]_{m,n}$ represents the channel coefficient between the 
$n$th transmit antenna at the BS and the 
$m$th reflecting element of the RIS.
The distance between the BS and the UAV is
\begin{equation}
d_b = \sqrt{\|\mathbf{q} - \mathbf{w}_b\|^2 + H^2}.
\end{equation}

As a worst-case assumption for secrecy analysis, we also consider a direct link between the BS and each of the UEHRs. 
The corresponding channel vector between the BS and the $j$-th UEHR is denoted by $\mathbf{h}_{bj} \in \mathbb{C}^{N_t \times 1}$ for all $j \in \mathcal{J}$, where small-scale fading follows the same Rician model as above, with $K_{bj}=0$ for NLoS. 
For simplicity, we assume a common path loss exponent for all air-to-ground links, i.e., $\alpha_b = \alpha_k = \alpha_j = \alpha_{\text{air}} = \alpha$.

To concurrently serve $K$ users, we adopt the RSMA technique throughout this paper. Accordingly, the transmitted signal $\mathbf{x} \in \mathbb{C}^{N_t \times 1}$ at the BS  is formulated as
\begin{equation}
\mathbf{x} = \mathbf{p}^c s_c + \sum_{k=1}^{K} \mathbf{p}_k^p s_k^p, \label{eq:tx_signal}
\end{equation}
where $\mathbf{p}^c$ denotes the precoding vector for the common stream and $\mathbf{p}_k^p$ is the precoding vector for the private stream of user $k$ . The symbols $s_c$ and $s_k^p$ represent the common and private data symbols, respectively.
Based on this transmission, the received signals at a legitimate user $k \in \mathcal{K}$ and at the $j$-th  UEHR, $j \in \mathcal{J}$,  are respectively given by:
\begin{align}
y_k &= \mathbf{g}_k^H \boldsymbol{\Theta} \mathbf{G}_b \mathbf{x} + n_k, \nonumber \\
y_j &= \left( \mathbf{h}_{bj}^H + \mathbf{h}_j^H \boldsymbol{\Theta} \mathbf{G}_b \right) \mathbf{x} + n_j, \label{eq:received_signals}
\end{align}
where $n_k$ and $n_j$ denote the additive white Gaussian noise (AWGN) at user $k$ and UEHR $j$, respectively.
Based on \eqref{eq:received_signals}, the received power at UEHR $j \in \mathcal{J}$ is expressed as
\begin{align}
&P_j^{\mathrm{EH}} = \left\vert \left( \mathbf{h}_{bj}^H + \mathbf{h}_j^H \boldsymbol{\Theta} \mathbf{G}_b \right) \mathbf{p}^c \right \vert^2 \nonumber \\
&\quad + \sum_{k=1}^{K} \left\vert \left( \mathbf{h}_{bj}^H + \mathbf{h}_j^H \boldsymbol{\Theta} \mathbf{G}_b \right) \mathbf{p}_k^p \right\vert^2. \label{eq:received_power}
\end{align}
Accordingly, the harvested power at UEHR $j$ can be modeled as \cite{Boshkovska}
\begin{equation}
\Omega(P_j^{\mathrm{EH}}) = \frac{\phi}{k'_1(1 + \exp(-b_0(P_j^{\mathrm{EH}} - b_1)))} - k'_2, \label{eq:harvested_power}
\end{equation}
where $\Omega(\cdot)$ is a logistic function representing the non-linear energy harvesting model, and $\phi$, $k'_1$, $k'_2$, $b_0$, $b_1$ are positive constants.
According to RSMA principles, the signal-to-interference-plus-noise ratio (SINR) of the common stream received by user $k \in \mathcal{K}$ and UEHR acting as an eavesdropper $j \in \mathcal{J}$ can be expressed as
\begin{align}
\gamma_k^{c} &= \frac{|\mathbf{g}_k^H \boldsymbol{\Theta} \mathbf{G}_b \mathbf{p}^c|^2}
{\sum\limits_{l=1}^{K} |\mathbf{g}_k^H \boldsymbol{\Theta} \mathbf{G}_b \mathbf{p}_l^p|^2 + \sigma^2}, \label{eq13} \\
\gamma_j^{c,e} &= \frac{|(\mathbf{h}_{bj}^H + \mathbf{h}_j^H \boldsymbol{\Theta} \mathbf{G}_b) \mathbf{p}^c|^2}
{\sum\limits_{k=1}^{K} |(\mathbf{h}_{bj}^H + \mathbf{h}_j^H \boldsymbol{\Theta} \mathbf{G}_b) \mathbf{p}_k^p|^2 + \sigma^2}. \label{eq14}
\end{align}
where $\sigma^2$ is the noise power. Similarly, for the private streams $s_k^p$, $k \in \mathcal{K}$, the SINR at user $k$ for decoding its own private stream is given by
\begin{align}
\gamma_k^{p}& = \frac{|\mathbf{g}_k^H \boldsymbol{\Theta} \mathbf{G}_b \mathbf{p}_k^p|^2}
{\sum\limits_{l \neq k} |\mathbf{g}_k^H \boldsymbol{\Theta} \mathbf{G}_b \mathbf{p}_l^p|^2 + \sigma^2}. \label{eq15}
\end{align}
We employ RSMA such that the UEHRs, regarded as potential eavesdroppers, are unable to decode the common stream. Furthermore, we consider the scenario where UEHR $j \in \mathcal{J}$ attempts to decode the private stream of user $k$. The corresponding SINR is expressed as
\begin{align}
    \gamma_{j,k}^{p,e} = \dfrac{|\mathbf{u}_j^H \mathbf{p}_k^p|^2}
{ |\mathbf{u}_j^H \mathbf{p}^c|^2 + \sum\limits_{l \neq k} |\mathbf{u}_j^H \mathbf{p}_l^p|^2 + \sigma^2 }. \label{eq:private_sinr_uehr}
\end{align}
where $\mathbf{u}_j$ denotes the combined channel between the BS, RIS, and UEHR $j$, defined as
$\mathbf{u}_j^H \triangleq \mathbf{h}_{bj}^H + \mathbf{h}_j^H \boldsymbol{\Theta} \mathbf{G}_b$.

The transmission rate $r_c$ of the common stream at the BS should be designed such that all legitimate users are able to decode it successfully, while also ensuring robustness against the UEHRs. Therefore, the rate $r_c$ must satisfy the following condition
\begin{equation}
 \log_2(1 + \max_{j \in \mathcal{J}}\gamma_j^{c,e})  \leq  r_c \leq \log_2(1 +\min_{k \in \mathcal{K}} \gamma_k^{c}). \label{eq:common_rate_constraint}
\end{equation}
By defining $\gamma^{c}_E \triangleq \max_{j \in \mathcal{J}} \gamma_j^{c,e}$ and 
$\gamma_{E}^{p,k} \triangleq \max_{j \in \mathcal{J}} \gamma_{j,k}^{p,e}$, 
the secrecy rate of user $k \in \mathcal{K}$ can be expressed as
\begin{align}
R_{k}^{\text{sec}}=a_{k}r_{c}^{\text{sec}}+r_{p,k}^{\text{sec}},\label{eq:secrecy_rate}
\end{align}
where 
\begin{align}
r_{c}^{\text{sec}} & =\left[\log_{2}(1+\gamma^{c})-\log_{2}(1+\gamma_{E}^{c})\right]^{+}\\
r_{p,k}^{\text{sec}} & =\left[\log_{2}(1+\gamma_{k}^{p})-\log_{2}(1+\gamma_{E}^{p,k})\right]^{+}
\end{align}
where $[x]^+ \triangleq \max(x, 0)$ and $\gamma^{c} \triangleq \min_{k \in \mathcal{K}} \gamma_k^{c}$.
Notably, $a_k$ adjusts the portion of the common message allocated to each user $k$, 
where $\sum_{k=1}^K a_k = 1$.

\section{Problem Formulation}\label{Problem-formulation}
To ensure perfect secrecy for all users, we define the minimum secrecy rate among all users as
\begin{equation}
R^{\text{sec}} = \min\nolimits_{k \in \mathcal{K}} R_k^{\text{sec}}. \label{eq:min_secrecy_rate}
\end{equation}
where $R^{\text{sec}}$ denotes the overall secrecy rate of the system.
Our goal is to maximize the SEE by jointly optimizing transmit precoders and passive beamforming. The problem can be mathematically formulated as
\begin{subequations} 
	\allowdisplaybreaks
\label{P1}
\begin{align}
\max_{\mathbf{a}, \mathbf{P}, \boldsymbol{\theta}} \quad & \frac{ R^{\text{sec}}}
{\varrho \text{Tr}(\mathbf{P} \mathbf{P}^H)+ P_0} \label{19a} \\
\text{s.t.} \quad 
& r_c \leq \min_{k \in \mathcal{K}} \log_2(1 + \gamma_k^{c}), \label{19c} \\
& \max_{j \in \mathcal{J}} \log_2(1 + \gamma_j^{c,e}) \leq r_c, \label{19d} \\
& \sum_{j \in \mathcal{J}} P_j^{\mathrm{EH}} \geq \Omega^{-1}(E_h),  \label{19e} \\
& \sum_{k=1}^K a_k = 1, \label{19g} \\
& 0 \leq a_k \leq 1, \quad \forall k \in \mathcal{K},  \label{19h} \\
& \text{Tr}(\mathbf{P} \mathbf{P}^H) \leq P_{\max},  \label{19i} \\
& \theta_m \in [0, 2\pi], \quad \forall m \in \{1, \dots, M\}, \label{19k}
\end{align}
\end{subequations}
where  $\mathbf{a} = \{ a_k \}_{k=1}^{K}$, $\mathbf{P} =  [\mathbf{p}^c, \mathbf{p}_1^p, \dots, \mathbf{p}_K^p]^T $, and $\boldsymbol{\theta} = \{ \theta_m \}_{m=1}^{M}$.
Moreover, $E_h$ denotes the minimum harvested energy requirement at the set of UEHRs, $\varrho$ is the reciprocal of the
drain efficiency of the power amplifier at the BS \cite{Hien}, and $P_0$ represents the circuit power consumption. The inverse function $\Omega^{-1}(x)$ is defined as:
\begin{equation}
\Omega^{-1}(x) = b_1 - \dfrac{1}{b_0} \ln\left( \frac{\phi}{k'_1(x + k'_2) }- 1  \right). \label{omegainv}
\end{equation}
The optimization problem in \eqref{P1} is nonconvex, making the global optimum difficult to obtain. To address this, we decompose it into subproblems and apply the SCA method.

\subsection{Optimal Allocation of Common Message Among IHRs}
For tractability, the problem in \eqref{P1} is solved using the block coordinate descent (BCD) method. Given the precoding matrix $\mathbf{P}$ and RIS phase-shift matrix $\boldsymbol{\Theta}$, the optimal allocation of the common message among the IHRs is determined by solving the following subproblem
\begin{subequations} 
	\allowdisplaybreaks
\label{eq:subproblem_allocation}
\begin{align}
\max_{\zeta, \mathbf{a}} \quad & \zeta \label{eq:subprob_obj} \\
\text{s.t.} \quad 
& \zeta \leq a_{k}r_{c}^{\text{sec}}+r_{p,k}^{\text{sec}},\ \forall k \in \mathcal{K}, \label{zeta-const.}\\
&  \eqref{19g}, \eqref{19h}. \label{eq:subprob_constraints}
\end{align}
\end{subequations}
where $\zeta$ is an auxiliary variable introduced to represent a lower bound of the minimum secrecy rate  $R^{\text{sec}}$. Problem \eqref{eq:subproblem_allocation} is a standard linear programming (LP) problem that can be efficiently solved using off-the-shelf convex optimization solvers.

\subsection{Optimization of Common and Private Precoders}
For simplicity of notation, we set 
$\varrho=1$ without loss of generality. The resulting subproblem becomes
\begin{subequations} 
	\allowdisplaybreaks
\label{eq:precoder_problem}
\begin{align}
\max_{\zeta, \mathbf{P}} \quad & \frac{\zeta}{\text{Tr}(\mathbf{P} \mathbf{P}^H)+ P_0}  \label{eq:precoder_obj} \\
\text{s.t.} \quad 
& \eqref{zeta-const.}, \eqref{19c}, \eqref{19d}, \eqref{19e}, \eqref{19i}. \label{eq:precoder_cons}
\end{align}
\end{subequations}
Problem~\eqref{eq:precoder_problem} is a nonlinear fractional program. To address its non-convexity, we apply Dinkelbach’s algorithm \cite{Dinkelbach}, introducing an auxiliary variable $\lambda$ to transform the objective into a more tractable form. The resulting problem is approximated as
\begin{subequations} 
\allowdisplaybreaks
\label{eq:precoder_transformed}
\begin{align}
\max_{\zeta, \mathbf{P}} \quad & \zeta - \lambda \left[ \text{Tr}(\mathbf{P} \mathbf{P}^H) + P_0 \right] \label{eq:precoder_lagrange} \\
\text{s.t.} \quad 
& \eqref{zeta-const.}, \eqref{19c}, \eqref{19d}, \eqref{19e}, \eqref{19i}. \label{eq:precoder_lagrange_cons}
\end{align}
\end{subequations}
The optimal value of $\lambda$ can be iteratively updated using Algorithm~\ref{Alg1} which is described later. However, Problem~\eqref{eq:precoder_transformed} remains nonconvex due to the nonconvex constraints in \eqref{eq:precoder_lagrange_cons}. To address this, we adopt an iterative approach that approximates the nonconvex terms using their first-order Taylor expansions at each iteration. In the following, we analyze each nonconvex term individually and derive its convex approximation.

We reformulate the non-convex constraints in \eqref{zeta-const.} using the epigraph technique. To this end, we define two sets of auxiliary variables
$\left\{ f_E^c, f_{E}^{p,k} \right\}$, and $\left\{ \rho^c, \rho_k^p\right\}$.

As a result, the secrecy constraint \eqref{zeta-const.} can be equivalently rewritten in terms of these auxiliary variables to enable convex optimization.
Using auxiliary variables and the first-order Taylor approximation, constraint~\eqref{zeta-const.} for all $k \in \mathcal{K}$ can be equivalently expressed as
\begin{subequations}\label{29}
	\allowdisplaybreaks
         \begin{align}
	&\zeta \leq a_k [\log_2(1+\rho^{c})-f_{E}^{c}]+[\log_2(1+\rho_{k}^{p})-f_{E}^{p,k}], \label{29a}
	\\&
 \log_2(1+\rho_{k}^{p})\geq f_E^{p,k},\label{29d}\\&
\log_2(1+\rho^{c}) \geq f_E^c, \label{29e}
 \\&
\rho^c \leq \min_{k \in \mathcal{K}} \gamma_k^{c},\label{29f}
\\&
\rho_k^p \leq \gamma_k^{p},\label{29g}
\\&
f_{E}^{c} \ge\log_2(1+\max_{j\in\mathcal{J}}\gamma_{j}^{c,e}),\label{29j}\\&
f_{E}^{p,k} \ge\log_2(1+\max_{j\in\mathcal{J}}\gamma_{j,k}^{p,e}).\label{29k}
\end{align} 
\end{subequations}
\begin{Proposition}
The affine approximation of constraint \eqref{29f} and  \eqref{29g}, $ \forall k \in \mathcal{K}$ are given by:
\begin{align}\label{27}
&\sum\limits_{\ell=1}^{K} | \mathbf{v}^H \mathbf{p}_\ell^p|^2 + \sigma^2 - \Psi^{(t)}(\mathbf{p}^{c},\rho^{c};\mathbf{v}) \leq 0\\ \label{28}
&\sum\limits_{\ell \neq k} | \mathbf{v}^H \mathbf{p}_\ell^p|^2 + \sigma^2 - \Psi^{(t)}(\mathbf{p}^p_k, \rho^p_k; \mathbf{v}) \leq 0,
\end{align}
where $\mathbf{v}^H \triangleq \mathbf{g}_k^H \boldsymbol{\Theta} \mathbf{G}_b$ and
$\Psi^{(t)}(\mathbf{u},x;\mathbf{h}) \triangleq 
\dfrac{2 \,\mathfrak{Re} \left\{ ( \mathbf{u}^{(t-1)})^H \mathbf{h} \mathbf{h}^H \mathbf{u} \right\} }{x^{(t-1)}} \;-\; \dfrac{\left | \mathbf{h}^H \mathbf{u}^{(t-1)} \right |^2 \, x}{\big(x^{(t-1)}\big)^2}$ is
the first-order Taylor approximations of  
$\dfrac{| \mathbf{h}^H  \mathbf{u}|^2}{x}$.
\end{Proposition}

\begin{proof}
Please refer to \cite{Hashempour}.
\end{proof}
To handle the nonconvexity of constraints \eqref{29j} and \eqref{29k}, we introduce auxiliary variables $\rho_{E}^{c}$ and $\rho^{p,k}_E$
satisfying
\begin{subequations}\label{26}
	\allowdisplaybreaks
         \begin{align}
& \rho_{E}^{c} \geq\max_{j\in\mathcal{J}}\gamma_{j}^{c,e}\\
&f_{E}^{c} \ge\log_2(1+\rho_{E}^{c}),\\&
\rho^{p,k}_E \geq \max_{j\in\mathcal{J}}\gamma_{j,k}^{p,e}\\&
f_{E}^{p,k} \ge\log_2(1+\rho^{p,k}_E).
\end{align} 
\end{subequations}
Subsequently, the affine approximations of the constraints in \eqref{26} are derived as follows
\begin{subequations}\label{31}
	\allowdisplaybreaks
         \begin{align}
    &\frac{| \mathbf{u}_j^H \mathbf{p}^c|^2}{\rho^c_E}-
    \sum_{k}  \Psi^{(t)}(\mathbf{p}^p_k, 1; \mathbf{u}_j) -\sigma^2\leq 0,\label{27a}
    \\
    &\frac{| \mathbf{u}_j^H \mathbf{p}^c|^2}{\rho^{p,k}_E }-\sum\limits_{\ell \neq k}
     \Psi^{(t)}(\mathbf{p}^p_\ell, 1; \mathbf{u}_j) -\Psi^{(t)}(\mathbf{p}^c, 1; \mathbf{u}_j) -\sigma^2\leq 0,
     \\
    & 1 - \Gamma^{(t)}(f_E^c) +\rho_E^c\leq 0,\label{31d}\\
    & 1 - \Gamma^{(t)}(f_E^{p,k}) +\rho_E^{p,k}\leq 0,\label{31e}
\end{align} 
\end{subequations}
where 
$\Psi^{(t)}(\mathbf{p}^p_k, \rho^p_k; \mathbf{v})$ 
and $\Psi^{(t)}(\mathbf{p}^p_k, 1; \mathbf{u}_j)$ are the first-order Taylor approximations of  
$\frac{| \mathbf{v}^H  \mathbf{p}^p_k |^2}{\rho^p_k}$ 
and $|\mathbf{u}_j^H \mathbf{p}_k^p|^2$, respectively. Moreover,
$
\Gamma^{(t)}(x)  \triangleq 2^{x^{(t-1)}} \big[1+\ln(2)\,(x-x^{(t-1)})\big]
$
is the first-order Taylor approximation of $2^x$.

Following the same approach as in Proposition 1, an affine approximation of constraint \eqref{19c} is given by
\begin{subequations}
	\allowdisplaybreaks
\label{32}
\begin{align} 
& r_c \leq \log_2(1 + \rho^c), \\  
& \eqref{27}.
\end{align}
\end{subequations}
For the eavesdroppers, constraint~\eqref{19d} is reformulated as
\begin{subequations}
	\allowdisplaybreaks
\label{eq:common_eve_reform}
\begin{align}
&1 - \Gamma^{(t)}(r_c) + \rho_E^c \leq 0, \label{33a} \\&
\eqref{27a}. \label{33b}
\end{align}
\end{subequations}
Finally, the energy harvesting constraint~\eqref{19e} can be reformulated as
\begin{equation} \label{eq:energy_reform}
\sum_{j \in \mathcal{J}} \Phi^{(t)}(\mathbf{p}^c, \mathbf{u}_j) + \sum_{j \in \mathcal{J}} \sum_{k} \Phi^{(t)}(\mathbf{p}_k^p, \mathbf{u}_j) \geq \Omega^{-1}(E_h),
\end{equation}
where $\Phi^{(t)}(\mathbf{a}, \mathbf{b})$ denotes the first-order Taylor approximation of $|\mathbf{a}^H \mathbf{b}|^2$ around $\mathbf{b}^{(t-1)}$, and is defined as
\begin{equation}
\Phi^{(t)}(\mathbf{a}, \mathbf{b}) = 2 \cdot \mathfrak{Re}\left\{ \mathbf{b}^{(t-1)H} \mathbf{a} \mathbf{a}^H \mathbf{b} \right\} - \left| \mathbf{a}^H \mathbf{b}^{(t-1)} \right|^2.
\end{equation}

With these approximations, the original precoder optimization problem~\eqref{eq:precoder_problem} can be reformulated into a convex problem and solved iteratively, as follows
\begin{align}\label{P2}
        \big( \zeta^{(t)}, \mathbf{P}^{(t)} \big) =
        \arg \max_{\zeta,\,\mathbf{P}}\
         \zeta - \lambda^{(t)}
        \big(
        \text{Tr}(\mathbf{P}\mathbf{P}^H) + P_0
        \big)  \nonumber\\
        \text{s.t. } \eqref{29a}–\eqref{29e},~\eqref{27}–\eqref{eq:energy_reform},~\eqref{19i}.
\end{align}
where $\lambda^{(t)}$ is updated as
\begin{align}\label{lambda_t}
    \lambda^{(t)} = \frac{\zeta^{(t-1)}}{
   \text{Tr}(\mathbf{P}^{(t-1)}\mathbf{P}^{(t-1)H}) + P_0 }.
\end{align}
The overall iterative procedure for solving problem~\eqref{eq:precoder_problem} is summarized in Algorithm~\ref{Alg1}.
\begin{algorithm}[t]
\caption{Iterative Algorithm for Problem~\eqref{eq:precoder_problem}}
\label{Alg1}
\KwIn{Initial values $\zeta^{(0)}$, $\mathbf{P}^{(0)}$, tolerance $\varepsilon$, $t = 1$.}

\Repeat{$\big| \zeta^{(t)} - \zeta^{(t-1)} \big| \leq \varepsilon$}{
    Update $\lambda^{(t)}$ using \eqref{lambda_t};
    
    Solve the convex subproblem \eqref{P2};
    
    Set $t \gets t+1$;
}

\Return $\mathbf{P}^{(t)}$\;
\end{algorithm}

\subsection{RIS Reflection Phase Optimization}

We first note that, given the precoder vectors, the power consumption of the system is fixed. Thus, the fairness energy efficiency maximization problem reduces to a fairness secrecy rate maximization problem. Accordingly, we can reformulate the following subproblem
\begin{align}
\max_{\boldsymbol{\Theta},\zeta} \quad & \zeta\label{eq:irs_opt_obj} \\
\text{s.t.} \quad 
& \text{\eqref{zeta-const.}, \eqref{19c}, \eqref{19d}, \eqref{19e}, \eqref{19k}}. 
\label{36}
\end{align}
For any arbitrary  vectors $\mathbf{v}$ and $\mathbf{w}$ we have: $\mathbf{v}^H \boldsymbol{\Theta}\mathbf{G}_b \mathbf{w} = \mathbf{t}^H \boldsymbol{s}$, where $\mathbf{t} = (\text{diag}(\mathbf{v}^H) \mathbf{G}_b \mathbf{w})^*$.
Therefore, by defining 
$\mathbf{t}_{kc} \triangleq (\text{diag}(\mathbf{g}_k^H) \mathbf{G}_b \mathbf{p}^c)^*$, 
$\mathbf{t}_{jc} \triangleq (\text{diag}(\mathbf{h}_j^H) \mathbf{G}_b \mathbf{p}^c)^*$, 
$\mathbf{t}_{kp} \triangleq (\text{diag}(\mathbf{g}_k^H) \mathbf{G}_b \mathbf{p}_k^p)^*$, 
$\mathbf{t}_{jp} \triangleq (\text{diag}(\mathbf{h}_j^H) \mathbf{G}_b \mathbf{p}_k^p)^*$,
$\mathbf{t}_{kp\ell} \triangleq (\text{diag}(\mathbf{g}_k^H) \mathbf{G}_b \mathbf{p}_\ell^p)^*$, 
$\mathbf{t}_{jp\ell} \triangleq (\text{diag}(\mathbf{h}_j^H) \mathbf{G}_b \mathbf{p}_\ell^p)^*$,
we proceed to convexify the constraints of the RIS phase optimization subproblem. 
According to \eqref{29} we reformulate the constraint \eqref{zeta-const.} using the following proposition
\begin{Proposition}
An affine approximation of constraint \eqref{zeta-const.}, $ \forall k \in \mathcal{K}$ and $ \forall j \in \mathcal{J}$ is given by:
\begin{subequations}\label{44}
	\allowdisplaybreaks
         \begin{align}
         &\sum_{\ell=1}^{K} |\mathbf{t}_{kp\ell}^H \mathbf{s}|^2 + \sigma^2 - \Psi^{(t)}(\mathbf{s}, \rho^c; \mathbf{t}_{kc}) \leq 0, \label{44a}\\
	&\sum\limits_{\ell \neq k} |\mathbf{t}_{kp\ell}^H \mathbf{s}|^2 + \sigma^2 - \Psi^{(t)}(\mathbf{s}, \rho^p_k;\mathbf{t}_{kp}) \leq 0, \label{44b}\\&
   |\mathbf{h}_{bj}^H \mathbf{p}^c + \mathbf{t}_{jc}^H \mathbf{s}|^2 \leq \varTheta^{(t)}( \xi_{jp}, \rho_E^c), \label{44c}\\&
   \xi_{jp} \leq \sum_k \vartheta^{(t)}(\mathbf{h}_{bj}^H \mathbf{p}_k^p, \mathbf{t}_{jp};\mathbf{s}) + \sigma^2,
   \label{44d}\\&
   |\mathbf{h}_{bj}^H \mathbf{p}_k^{p} + \mathbf{t}_{jp}^H \mathbf{s}|^2 \leq \varTheta^{(t)}( \xi_{jc\ell}, \rho^{p,k}_E), \label{44e}\\&
   \xi_{jc\ell} \leq \sum\limits_{\ell \neq k} \vartheta^{(t)}(\mathbf{h}_{bj}^H \mathbf{p}^{p}_\ell, \mathbf{t}_{jp\ell};\mathbf{s}) + \vartheta^{(t)}(\mathbf{h}_{bj}^H \mathbf{p}^c, \mathbf{t}_{jc};\mathbf{s}) + \sigma^2,\label{44f}
   \\&
   \eqref{29a}-\eqref{29e},\eqref{31d},\eqref{31e}. \label{44g}
\end{align} 
\end{subequations}
where we have defined
$\varTheta^{(t)}(x,y) \triangleq 
 \tfrac{1}{2} \big(x^{(t-1)}+y^{(t-1)}\big)(x+y)-\tfrac{1}{4}\big(x^{(t-1)}+y^{(t-1)}\big)^2-\tfrac{1}{4}(x-y)^2,$
for the linear approximation of the terms that involve a product of two variables. Additionally, $\vartheta^{(t)}(c; \mathbf{t}; \mathbf{s})$ is the first-order Taylor approximation of $|c + \mathbf{t}^H \mathbf{s}|^2$ at iteration $t$:
\begin{equation}
\vartheta^{(t)}(c; \mathbf{t}; \mathbf{s}) = 2 \, \mathfrak{Re}\left\{ \big(c + \mathbf{t}^H \mathbf{s}^{(t-1)}\big)^H \mathbf{t}^H \mathbf{s} \right\} - \big|c + \mathbf{t}^H \mathbf{s}^{(t-1)}\big|^2.
\end{equation}
\end{Proposition}

\begin{proof}
Please refer to \cite{Hashempour}.
\end{proof}
From the remaining constraints in problem \eqref{36}, we replace \eqref{19c} with its affine approximation in \eqref{32}. 
Constraint \eqref{19d} is approximated by \eqref{33a}, \eqref{44c}, and \eqref{44d}.
Moreover, \eqref{19e} can be rewritten as
\begin{align}
&\sum_{j \in \mathcal{J}} \vartheta^{(t)}(\mathbf{h}_{bj}^H \mathbf{p}^c; \mathbf{t}_{jc}; \mathbf{s}) + \sum_{j \in \mathcal{J}} \sum_{k} \vartheta^{(t)}(\mathbf{h}_{bj}^H \mathbf{p}_k^p; \mathbf{t}_{jp}; \mathbf{s}) \nonumber \\& \geq \Omega^{-1}(E_h). \label{48}
\end{align}
Constraint \eqref{19k} is equivalent to 
\begin{align}\label{49}
|s_m| = 1, \quad \forall m \in \{1, \dots, M\}.
\end{align}
However, \eqref{49} is still a non-convex constraint. Hence, following
the penalty-based method in \cite{Kumar}, we handle the unit-modulus
constraint and rewrite the objective function as follows
\begin{subequations}\label{50}
	\allowdisplaybreaks
\begin{align}
\max_{\zeta,\, \mathbf{s}} \quad & \zeta + C  \|\mathbf{s}\|^2  \\
\text{s.t.} \quad & \eqref{44}, \eqref{32}, \eqref{33a}, \eqref{44c}, \eqref{44d}, \eqref{48}, \\ &   |\mathbf{s}| \le \mathbf{1},  \label{50c}
\end{align}
\end{subequations}
where $C$ is a large positive constant, $|\mathbf{s}|$ denotes the element-wise magnitude of 
$\mathbf{s}$, $\mathbf{1}$ is a vector of ones, and the inequality in \eqref{50c} is understood entry-wise.
Note that the penalty term enforces $|s_m|^2 = 1$ at the optimal solution.
We further approximate the non-convex part, i.e., the penalty term in the objective of \eqref{50}.
Toward this end, we directly linearize the quadratic term $|s_m|^2$ around the point $s_m^{(t-1)}$ using its first-order Taylor expansion.
Accordingly, the penalty term can be approximated as $C \sum_{m=1}^M \Big( 2\,\mathfrak{Re}\{ (s_m^{(t-1)})^{*} s_m \} - |s_m^{(t-1)}|^2 \Big)$.
Thus, the problem can be rewritten as
\begin{subequations}\label{51}
	\allowdisplaybreaks
\begin{align}
\max_{\zeta,\, \mathbf{s}} \quad & \zeta 
+ C \left( 2\,\mathfrak{Re}\big\{ (\mathbf{s}^{(t-1)})^{H} \mathbf{s} \big\} 
- \|\mathbf{s}^{(t-1)}\|^2 \right) 
\\
\mathrm{s.t.} \quad & \eqref{44}, \eqref{32}, \eqref{33a}, \eqref{44c}, \eqref{44d}, \eqref{48}, \eqref{50c}.
\end{align}
\end{subequations}
Here, $(t-1)$ denotes the previous iteration of the phase optimization subproblem in \eqref{51}.
Finally, the RIS phase optimization can be summarized in Algorithm~\ref{Alg2}.
\begin{algorithm}[t]
\caption{Iterative Algorithm for Phase Optimization Problem}
\label{Alg2}
\KwIn{Initial values $\zeta^{(0)}$, $\boldsymbol{\Theta}^{(0)}$, $\mathbf{s}^{(0)}$, tolerance $\varepsilon$.}
\KwOut{Optimized phase shifts $\boldsymbol{\Theta}^*$.}

Set $t \gets 1$\;

\Repeat{$|\zeta^{(t)} - \zeta^{(t-1)}| \le \epsilon$}{
    Solve the optimization problem in \eqref{51}\;
    $t \gets t+1$\;
}

\Return $\boldsymbol{\Theta}^{(t)}$\;
\end{algorithm}
By sequentially optimizing ($\mathbf{a}, \mathbf{P}, \boldsymbol{\theta}$), the joint optimization of the common message coefficient, BS active beamforming, and RIS phase shifts  can be effectively solved using an AO scheme. The complete procedure is presented in Algorithm \ref{alg:overall}.
\begin{algorithm}[t]
\caption{Proposed Algorithm for Fairness Secrecy Energy Efficiency Maximization}
\label{alg:overall}
\DontPrintSemicolon
\LinesNumbered
\KwIn{Initialize $\mathbf{P}^{(0)},  \boldsymbol{\Theta}^{(0)}, \eta^{(0)}$, tolerance $\varepsilon$ for Algorithm~\ref{Alg1}, iteration index $i \gets 0$, and convergence accuracy $\epsilon$.}
\Repeat{$\big|\eta^{(i)} - \eta^{(i-1)}\big| \le \epsilon$}{
  Obtain message allocation ratios $\mathbf{a}^{(i+1)}$ by solving \eqref{eq:subproblem_allocation} with given $\big(\mathbf{P}^{(i)}, \boldsymbol{\Theta}^{(i)}\big)$. \;
  Obtain RS precoders $\mathbf{P}^{(i+1)}$ via Algorithm~\ref{Alg1} with given $\big(\mathbf{a}^{(i+1)}, \boldsymbol{\Theta}^{(i)}\big)$. \;
  Obtain RIS phase shifts $\boldsymbol{\Theta}^{(i+1)}$ via Algorithm~\ref{Alg2} with given $\big(\mathbf{a}^{(i+1)}, \mathbf{P}^{(i+1)}\big)$. \;
  Compute
  $\eta^{(i+1)} = \dfrac{\zeta^{(i+1)}}{\sum_{k=1}^{K}\|\mathbf{p}_k^{p\, (i+1)}\|^2 + \|\mathbf{p}^{c\, (i+1)}\|^2 + P_0}$. \;
  \If{$\eta^{(i+1)} < \eta^{(i)}$}{
    Set  $\boldsymbol{\Theta}^{(i+1)} \gets \boldsymbol{\Theta}^{(i)}$. \;
  }
  $i \gets i + 1$\;
}
\end{algorithm}

\section{Simulation Results}\label{Simulat}
Unless otherwise stated, the setup is as follows. 
We consider a $1500 \times 3000 \,\text{m}^2$ area with the BS at $(0,0)$ and the UAV at $(1000,0,100)\,\text{m}$. 
The BS employs $N_t=4$ antennas, serving $K=2$ legitimate users and $J=2$ UEHRs, randomly distributed within circles of radius $500$~m centered at $(-1000,1000)$~m and $(-1000,-1000)$~m, respectively.
The UAV carries an RIS with $M=16$ reflecting elements. 
Other parameters are: path-loss exponent $\alpha=2.5$, circuit power $P_0=1$~W, and minimum harvested energy $E_h=0.01$~J. 
Noise power is normalized to 0~dBm and the minimum required transmission rate of the common stream is $r_c=0.5$ Bit/sec/Hz \cite{Hashempour}. 
All results are averaged over 100 channel realizations with convergence accuracy $\epsilon=\varepsilon=0.01$.

Figure~\ref{fig3} compares the SEE versus $N_t$ for the proposed RSMA scheme against SDMA and NOMA baselines, which are implemented as special cases of RSMA following the modeling approach in \cite{Bastami}, with $P_{\max}=10$ dBm.
SEE increases with $N_t$ and then saturates, since extra antennas provide mainly array gains while secrecy rate grows only logarithmically. 
RSMA achieves the highest SEE by flexibly managing interference. 
SDMA slightly outperforms NOMA, as beamforming is more effective than SIC in this scenario.
\begin{figure} 
	\centering
\includegraphics[width=.85\linewidth]{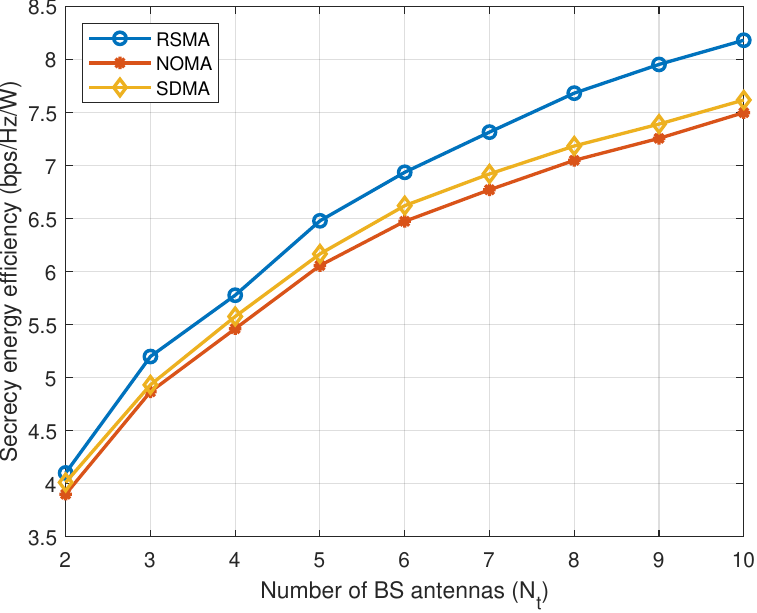}
	\caption{SEE versus $N_t$ for $P_{\max}=10$ dBm and $M=16$ compared with NOMA and SDMA benchmarks.}
	\label{fig3}		
\end{figure}
\begin{figure}[t] 
	\centering
\includegraphics[width=.85\linewidth]{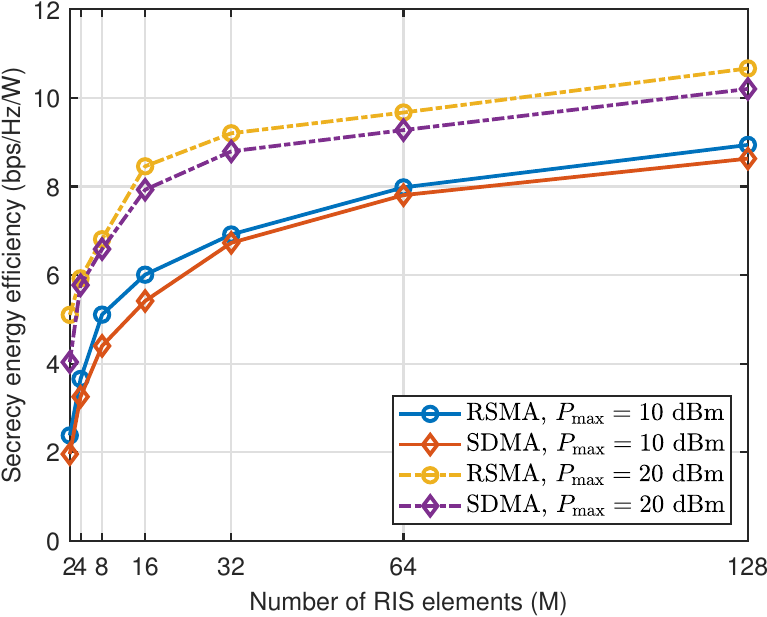}
	\caption{SEE versus $M$ for $N_t=4$ and different schemes.}
	\label{fig4}		
\end{figure}
Figure~\ref{fig4} shows the impact of RIS size $M$. 
Increasing $M$ enhances SEE with a logarithmic trend, and higher $P_{\max}$ further boosts performance. 
Notably, $M=128$ with $P_{\max}=10$ dBm achieves almost the same SEE as $M=16$ with $P_{\max}=20$ dBm, demonstrating the RIS’s energy-saving potential.
\begin{figure}[t] 
	\centering
\includegraphics[width=.85\linewidth]{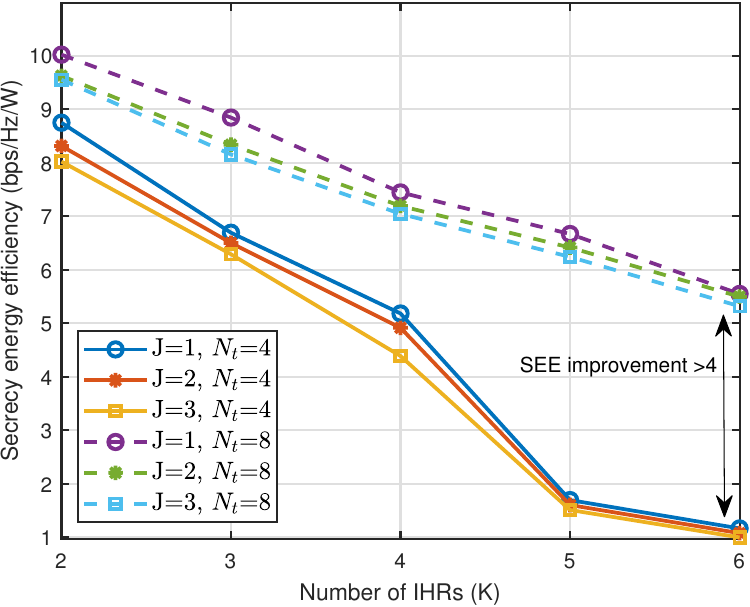}
	\caption{SEE versus number of legitimate users ($K$) and UEHRs ($J$) for $P_{\max}=20$ dBm.}
	\label{fig5}		
\end{figure}
Finally, Fig.~\ref{fig5} evaluates scalability. 
As $K$ increases, SEE decreases significantly due to higher eavesdropping probability, while more UEHRs cause only a slight drop. 
Increasing $N_t$ improves spatial diversity and offsets the degradation. 
For instance, at $K=6$, raising $N_t$ from 4 to 8 yields more than a fourfold gain in SEE.

\section{Conclusions}\label{conc}

We investigated a UAV-assisted RIS-RSMA system for secure and energy-efficient downlink transmission in the presence of UEHRs. 
An AO-based algorithm was proposed to jointly optimize message allocation, active precoders, and RIS phase shifts under power and EH constraints. 
Simulation results verified fast convergence and demonstrated that the proposed scheme achieves higher SEE compared to SDMA and NOMA baselines, highlighting the strong synergy between RSMA and RIS in UAV-assisted secure networks.


\end{document}